\documentclass[runningheads]{llncs}

\usepackage[export]{adjustbox}
\usepackage{array}
\usepackage{arydshln}
\usepackage{breakurl} 
\usepackage{amssymb}
\usepackage{amsfonts}
\usepackage{amsmath}
\usepackage{caption}
\usepackage{cite}
\usepackage{color}
\usepackage{commath}
\usepackage{enumitem}

\usepackage{float}
\usepackage{framed}
\usepackage{graphicx}
\usepackage[breaklinks]{hyperref}
\usepackage{mathtools}
\usepackage{multirow}
\usepackage{nameref}
\usepackage{placeins}
\usepackage{stackengine}
\usepackage{tabularx}
\usepackage{xspace}
\usepackage{xparse}
\usepackage{url}
\usepackage{listings}

\usepackage{algorithm}
\usepackage[noend]{algpseudocode}
\usepackage{etoolbox}\AtBeginEnvironment{algorithmic}{\footnotesize}

\definecolor{mygreen}{rgb}{0,0.6,0}
\definecolor{mygray}{rgb}{0.5,0.5,0.5}
\definecolor{mymauve}{rgb}{0.58,0,0.82}

\hypersetup{
  colorlinks   = true, 
  urlcolor     = blue, 
  linkcolor    = black, 
  citecolor   = red 
}

\usepackage[nonumberlist,acronyms,nogroupskip]{glossaries}

\newglossary*{not}{Notation}
\usepackage{doi}

\makenoidxglossaries
\setacronymstyle{long-short}

\newcommand*{\pcr}[1]{{\fontfamily{pcr}\selectfont #1}}
\newcommand{\circler}{\textsuperscript{\textregistered}}

\newtheorem{observation}{Observation}

\newcommand{\Z}{\mathbb{Z}}

\newcommand{\N}{\mathbb{N}}
\newcommand{\A}{\mathcal{A}}
\newcommand{\G}{\mathbb{G}}

\newcommand{\alice}{\ensuremath{s}\xspace}
\newcommand{\bob}{\ensuremath{r}\xspace}
\newcommand{\aliceElement}{\ensuremath{\alice}\xspace}
\newcommand{\bobElement}{\ensuremath{\bob}\xspace}

\newcommand{\aliceNRecords}{\ensuremath{N_{\alice}}\xspace}
\newcommand{\bobNRecords}{\ensuremath{N_{\bob}}\xspace}

\newcommand{\db}[1]{\ensuremath{D_{#1}}\xspace}
\newcommand{\party}[1]{\ensuremath{P_{#1}}\xspace}

\newcommand{\measure}[2]{\ensuremath{\mu(#1, #2)}\xspace}
\newcommand{\measureE}{\ensuremath{\mu}\xspace}

\newcommand{\indicator}[4]{\ensuremath{I^{#1}_{#2}(#3, #4)}\xspace}
\newcommand{\indicatorE}[2]{\ensuremath{I^{#1}_{#2}}\xspace}

\newcommand{\lshMatch}[2]{\ensuremath{\texttt{LSHMatch}(#1, #2)}\xspace}
\newcommand{\lshMatchE}{\ensuremath{\texttt{LSHMatch}}\xspace}

\newcommand{\recordSpace}{\ensuremath{\mathcal{R}}\xspace}

\newcommand{\jaccard}[2]{\ensuremath{J(#1, #2)}\xspace}

\DeclarePairedDelimiterXPP\E[1]{\mathbb{E}}{[}{]}{}{

#1
}

\usepackage{tikz}
\usetikzlibrary{svg.path}
\definecolor{orcidlogocol}{HTML}{A6CE39}
\tikzset{
  orcidlogo/.pic={
    \fill[orcidlogocol] svg{M256,128c0,70.7-57.3,128-128,128C57.3,256,0,198.7,0,128C0,57.3,57.3,0,128,0C198.7,0,256,57.3,256,128z};
    \fill[white] svg{M86.3,186.2H70.9V79.1h15.4v48.4V186.2z}
                 svg{M108.9,79.1h41.6c39.6,0,57,28.3,57,53.6c0,27.5-21.5,53.6-56.8,53.6h-41.8V79.1z M124.3,172.4h24.5c34.9,0,42.9-26.5,42.9-39.7c0-21.5-13.7-39.7-43.7-39.7h-23.7V172.4z}
                 svg{M88.7,56.8c0,5.5-4.5,10.1-10.1,10.1c-5.6,0-10.1-4.6-10.1-10.1c0-5.6,4.5-10.1,10.1-10.1C84.2,46.7,88.7,51.3,88.7,56.8z};
  }
}
\newcommand{\OrigHeightRecip}{0.00390625}
\newlength{\curXheight}
\DeclareRobustCommand\orcid[1]{%
\texorpdfstring{%
\setlength{\curXheight}{\fontcharht\font`X}%
\href{https://orcid.org/#1}{\XeTeXLinkBox{\mbox{%
\begin{tikzpicture}[yscale=-\OrigHeightRecip*\curXheight,
xscale=\OrigHeightRecip*\curXheight,transform shape]
\pic{orcidlogo};
\end{tikzpicture}%
}}}}{}}

\renewcommand{\paragraph}[1]{{\smallskip\noindent {\bf{#1}}~}}

\newacronym{BDL}{BDL}{blind data linkage}
\newacronym{BFE}{BFE}{bloom filter encodings}
\newacronym{DH}{DH}{Diffie-Hellmann}
\newacronym{ER}{ER}{Entity resolution}
\newacronym{HE}{HE}{homomorphic encryption}
\newacronym{LAN}{LAN}{local area network}
\newacronym{WAN}{WAN}{wide area networks}
\newacronym{LSH}{LSH}{local sensitive hash}
\newacronym{LPH}{LPH}{locality preserving hash}
\newacronym{SSN}{SSN}{social security number}
\newacronym{OT}{OT}{oblivious transfer}
\newacronym{PPRL}{PPRL}{privacy-preserving record linkage}
\newacronym{PSI}{PSI}{private set intersection}
\newacronym{PPT}{PPT}{probabilistic polynomial-time}
\newacronym{RL}{RL}{record linkage}
\newacronym{QID}{QID}{quasi-identifier}

\begin{document}

\author{%
Allon Adir\inst{1}\orcid{0000-0001-8128-6706} \and 
Ehud Aharoni\inst{1}\orcid{0000-0002-3647-1440} \and
Nir Drucker\inst{1}\orcid{0000-0002-7273-4797} \and
Eyal Kushnir\inst{1}\orcid{0000-0001-6123-0297}  \and
Ramy Masalha\inst{1}\orcid{0000-0002-6808-5675}  \and
Michael Mirkin\inst{2}\thanks{The work for this paper was done while Michael Mirkin was with IBM Research.}\orcid{0000-0002-7332-7667} \and
Omri Soceanu\inst{1}\orcid{0000-0002-7570-4366}}

\authorrunning{ }%

\institute{$^1$IBM Research, Haifa, Israel} 
\institute{IBM Research - Haifa \and
Technion - Israel Institute of Technology}

\title{Privacy-preserving record linkage using local sensitive hash and private set intersection}

\titlerunning{LSH-PSI PPRL}

\maketitle

\begin{abstract}
The amount of data stored in data repositories increases every year. This makes it challenging to link records between different datasets across companies and even internally, while adhering to privacy regulations. Address or name changes, and even different spelling used for entity  data, can prevent companies from using private deduplication or record-linking solutions such as private set intersection (PSI). 
To this end, we propose a new and efficient privacy-preserving record linkage (PPRL) protocol that combines PSI and local sensitive hash (LSH) functions, and runs in linear time. We explain the privacy guarantees that our protocol provides and demonstrate its practicality by executing the protocol over two datasets with $2^{20}$ records each, in $11-45$ minutes, depending on network settings.

\end{abstract}

\section{Introduction}\label{sec:intro}
\gls{ER} is the process of identifying similar entities in several datasets, where the datasets may belong to different organizations. While these organizations would like to join hands and analyzes the behavior of matching customers, they may be restricted by law from sharing sensitive client-data such as medical, criminal, or financial information. The problem of matching records in two or more datasets without revealing additional information is called \gls{PPRL} \cite{PPRL-old} or \gls{BDL} \cite{old-2} and is the focus of this paper. A survey of \gls{PPRL} methods is available in \cite{pprl-survey21}. 
The importance of finding efficient and accurate \gls{PPRL} solutions can be observed, for example, in the establishment of a special task team by the Interdisciplinary Committee of the International Rare Diseases Research Consortium (IRDiRC) to explore different \gls{PPRL} approaches \cite{task-force}.

The \gls{PPRL} problem is a generalization of the well-studied \gls{PSI} problem in which two parties with different datasets would like to know the intersection or the size of the intersection of these datasets without revealing anything else about their data to the other party. Examples for \gls{PSI} solutions include \cite{psi-orig, psi-dh, psi-ot-benny, psi-he-unbalanced, psi-he, labled-psi}. With \gls{PSI}, the two parties compute the intersection of their respective sets, which can be used to identify matches by looking for records that share the same identifying field e.g., \gls{PSI} over \glspl{SSN}. However, in reality, such identifying fields do not always exist, and even when they do exist, their content may be entered incorrectly or differently. For example, consider two parties that perform \gls{PSI} on entity names. A single user may register himself in different systems under the names: `John doe', `John P Doe', `john doe', just `John', or even `Jon ode' by mistake. A general \gls{PPRL} solution may attempt to consider all of the above names as matching. 

In some cases, more than one data field is used to match two records, e.g., first name, last name, addresses, and dates of birth. These fields are known as \glspl{QID}, which may hold private information. In this paper, we assume that the parties are allowed to learn data by matching \glspl{QID}. In other cases, one can use a masking method e.g., as in \cite{Kargupta2005} to maintain the users' privacy. 

Non-exact matching is commonly performed using \gls{ER} solutions that employ a \gls{LSH} function. Unlike cryptographic hash functions, this technique permits collisions by deliberately hashing similar inputs to a single digest. For example, consider a hash function that hashes all the above names to a single digest value or to lists of digests with non-empty intersection. Different \gls{LSH} functions with different parameters allow us to fine-tune the results in different ways. We provide more details in Section \ref{sec:LSH}. 

Unfortunately, few practical protocols exist that can securely perform such ``fuzzy'' record linkage without revealing some private data of the parties, and do so in a linear time frame. See Section \ref{sec:rel} for a review of the different approaches. Many involve a third-party (e.g., \cite{federal}), which we aim to avoid, while other works do not provide a thorough leakage analysis that would help  evaluate the security of the solution. To this end, we constructed a new and efficient \gls{PPRL} solution that runs in $\mathcal{O}(n)$. We describe its performance and discuss its security characteristics. 

The goal of our solution is to compose a \gls{PSI} with an \gls{LSH} function. The dataset fields are first locally hashed by both parties using the \gls{LSH} and then checked for matches using \gls{PSI}. The choice of \gls{PSI} algorithm can only affect the performance (latency and bandwidth) of our solution but does not affect the amount of leaked information that can be tuned using the different parameters of the \gls{LSH}. Figure \ref{fig:illustration} illustrates a high-level view of our solution. 

\paragraph{Our contribution.}
Our contributions can be summarized as follows:
\begin{itemize}
    \item We introduce a novel and efficient \gls{PPRL} protocol that combines \gls{LSH} and \gls{PSI}, and analyze its security against semi-honest adversaries. It does not involve third parties. Specifically, due to the use of \gls{LSH}, our protocol has a low probability of revealing the data of non-matched records and thereby provides better privacy guarantees.
    \item We implemented the model and suggest several lower-level and higher-level optimizations.
    \item We evaluated our implementation over a dataset with $2^{20}$ records and demonstrated its practical advantage when the execution took $11-45$ minutes, depending on network settings.
    \item Our program is freely available for testing at \cite{helayers}.
\end{itemize}

\begin{figure}[h!]
    \centering
    \includegraphics[width=0.8\linewidth]{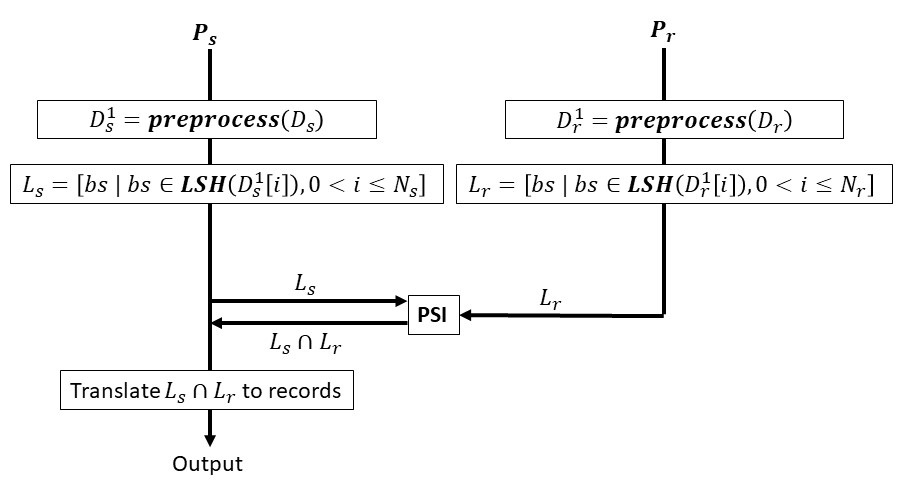}
    \caption{A high level illustration of our \gls{PPRL} protocol. The parties \party{\alice} and \party{\bob} hold datasets \db{\alice} and \db{\bob}. They preprocess the data for every record and then feed the results into an \gls{LSH} that outputs an ordered list of digest vectors $L_{\alice}$ and $L_{\bob}$, respectively. These are fed into a \gls{PSI} black box. Finally, \party{\alice} translates the \gls{PSI} output to the matching record IDs.}
    \label{fig:illustration}
\end{figure}

\subsection{Related work}\label{sec:rel}

To demonstrate our solution, we use a \gls{PSI} instantiation that uses public-key cryptography; specifically, we use one that leverages the commutative properties of the \gls{DH} key agreement scheme. This \gls{PSI} construction was introduced in \cite{psi-dh} with a similar construction even before that in \cite{psi-orig}. Subsequent \gls{PSI} works consider other, more complex cryptographic primitives such as \gls{HE} \cite{psi-he} and \gls{OT} \cite{psi-ot-benny}. While the latter solutions may offer an interesting tradeoff in terms of performance and security, we decided to stick with the basic \gls{DH}-style protocol due to its simplicity and the fact that its primitives were already standardized \cite{fips-80056b}. 
Because we use \gls{PSI} as a blackbox, we can also benefit from most of the advantages that the other methods provide such as performance and security guarantees. 

Our solution follows previous works in considering a \textit{balanced} case, where the two datasets are roughly equal in size. An example, for a \gls{PSI} over unbalanced sets was studied in \cite{psi-he-unbalanced}. In fact, there were attempts to use \gls{PSI} for \gls{PPRL} before this paper. However, they were either noted to be inefficient \cite{Vatsalan2017} or relied on a different techniques such as term frequency–inverse document frequency (TF-IDF) \cite{ravikumar2004secure}, which is more appropriate for comparing documents, rather than short record fields (such as names or addresses). Furthermore, the protocol of \cite{ravikumar2004secure} can only compare given record pairs. This implies the need for $\mathcal{O}(n^2)$ operations, in contrast to our method, which requires $\mathcal{O}(n)$ operations.

A complete survey of \gls{PPRL} techniques and challenges is available at \cite{Vatsalan2017,pprl-survey21}, in which we observed solutions that use different cryptographic primitives. For example, \cite{WONG20131280, Essex-HE} 
relies on \gls{HE}, which is known for its high computational cost. For example, \cite{Essex-HE} reports that it took somewhat less than two hours to evaluate $20,000$ patient records, which is less records than in our evaluations by several orders of magnitude. Other works \cite{GarbledC, SALEEM20181} use garbled circuits, which can still be inefficient, while other multi-party computation solutions such as \cite{SFour} can incur high communication costs \cite{chen2020current}. Another example is the fuzzy volts approach, which uses secure polynomial interpolations \cite{PPRL3}, but only reports results for around $1,000$ records. Other solutions \cite{dp-pplr2, dp-non-crypto} overcome the leakage issue by using differential privacy, which anonymizes the data to maintain privacy. We see it as an orthogonal approach to ours.

Many \gls{PPRL} works use Bloom filter encodings \cite{Schnell2009}, which use a \gls{LPH} function over the data. The main advantage of the Bloom filter is speed. The difference between \gls{LPH} and \gls{LSH} is that \gls{LPH} is data-dependent, i.e., for three records $p,q,r$, a metric $d$, and an \gls{LPH} function $l_p$
\[
d(p, q) < d(q,r) \Longrightarrow d(l_p(p), l_p(q)) < d(l_p(q),l_p(r))
\]
This relation complicates the evaluation of the protocol leakage. The lack of a formal analysis for Bloom filter based solutions caused several attacks on them \cite{cryptanalysis-bf, cryptanalysis-bf2, cryptanalysis-bf1, cryptanalysis-bf3}. A survey of attacks and countermeasures for this method can be found in \cite{Franke2021}. Our solution's use of \gls{LSH} has an advantage over Bloom filters as it is data-independent and more robust against the above attacks. A method that combines Bloom filters and \gls{LSH} was presented in \cite{Franke2018, Karapiperis2014}. In contrast to this one, our solution only uses \gls{LSH}, which simplifies the privacy analysis. Moreover, our use of \gls{PSI} hides the \gls{LSH} output and thus prevents offline attacks. In addition, \cite{Karapiperis2014} requires use of a third-party and demonstrates a solution that took more than an hour to match $300K$ records. Another recent example is \cite{SFour}, which runs in $\mathcal{O}(n \cdot polylog(n))$ and proved to be cryptographically secure in the semi-honest security model. However, the method analyzed $4,096$ records in $88$ minutes and it is not clear whether this method can scale to handle more than $100K$ records. 

\paragraph{Organization.}
The paper is organized as follows. Section \ref{sec:pre} provides some background notation and describes the required preliminaries for this work. In Section \ref{sec:pprl} we present and discuss several possible definitions of \gls{PPRL} protocols.
We provide a high level description of our solution in Section \ref{sec:solution} and provide further details about our implementation in Appendix \ref{sec:impl}. We report our experimental setup and results in Section \ref{sec:expr} and conclude in Section \ref{sec:conc}.

\section{Preliminaries and notation}\label{sec:pre}
We denote the concatenation of two strings by $s_1 ~|~ s_2$. The function $Eq(s,r)$ returns 1 when two strings are equal and 0 otherwise. An ordered list of elements $A$ is marked with square brackets, e.g., $A=[5,3,8]$ and we access its $i$th element by $A[i]$. A permutation $\pi$ can either return a permuted list when operating on an ordered list, or the index of a permuted element within that list when the input is another index. For example, let $\pi:x \mapsto x + 1 \pmod{4}$ be a permutation, then $\pi([5,6,7,8]) = [8, 5, 6, 7]$, $\pi(2)=3$, and $\pi(3)=0$. Uniform random sampling from a set $U$ is denoted by $u \xleftarrow{\$} U$.

\subsection{Entity resolution (ER)}\label{sec:ER}
An \gls{ER} method gets as input two datasets of $\aliceNRecords$ and $\bobNRecords$ records from record spaces $\recordSpace$: ${\db{\alice} = \{\aliceElement_1, \aliceElement_2, \dots, \aliceElement_{\aliceNRecords}\}}$ and ${\db{\bob} = \{\bobElement_1, \bobElement_2, \dots, \bobElement_{\bobNRecords}\}}$, respectively. It evaluates the similarity of every two records using a similarity measure $\mu : \recordSpace \times \recordSpace \rightarrow [0,1]$ and an associated similarity indicator
\begin{align*}
I^\mu_t : \recordSpace \times \recordSpace & \longrightarrow \{0,1\} \\
(\aliceElement, \bobElement) & \longmapsto 
\begin{cases}
    1 & \measure{\aliceElement}{\bobElement} \ge t \\
    0 & otherwise
\end{cases}
\end{align*}

The \gls{ER} method uses the similarity indicator to facilitate a bipartite graph $G=(U, V, E)$, where the nodes of $U$, $V$ are the records of \db{\alice}, \db{\bob}, respectively, and for every two nodes ($u \in U$, $v \in V$), an edge exists in $E$ if \indicator{\mu}{t}{u}{v}=1.

\paragraph{PPRL.} Informally, a \gls{PPRL} protocol is an \gls{ER} method executed by two parties: a sender \party{\alice} and a receiver \party{\bob}, who privately hold \db{\alice} and \db{\bob}, respectively. At the end of the protocol, \party{\bob} learns the similarity edges $E$ while \party{\alice} learns nothing. We provide a formal definition in Section \ref{sec:pprl}. Specifically, our \gls{PPRL} solution uses the \gls{LSH} and \gls{PSI} primitives, described next.

\subsection{Local sensitive hash (LSH)}\label{sec:LSH}
An \gls{LSH} \cite{ullman} is a hash function that deliberately hashes \emph{similar} inputs to the same output hash value. We are interested in the similarity of strings i.e., the content of the record fields. Therefore, we use the \gls{LSH} from \cite{ullman}, which is based on the Jaccard index and on \textit{Min-Hashes}, as demonstrated in Figure \ref{fig:LSH}.

\begin{figure}[ht!]
    \centering
    \includegraphics[width=0.7\linewidth]{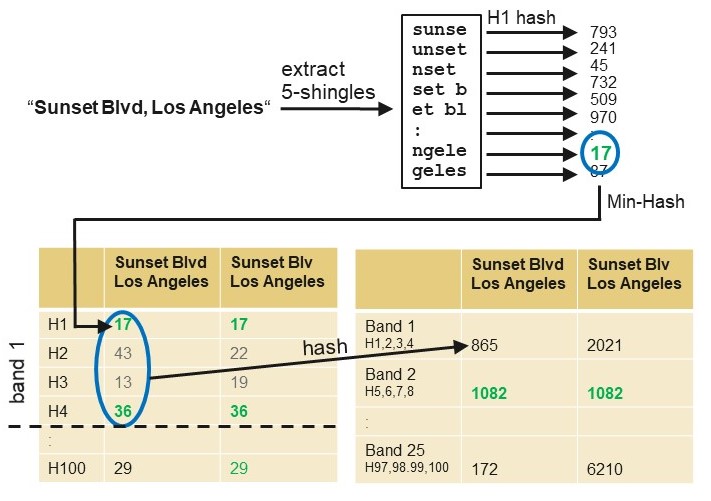}
    \caption{Computing the \gls{LSH} for a string: shingles are extracted from the normalized string, and then min-hashes are evaluated and grouped into bands that are hashed to a list of signatures.}
    \label{fig:LSH}
\end{figure}

\emph{Jaccard index} (a.k.a. the Jaccard similarity coefficient) is a similarity measure for strings. The procedure for computing the Jaccard index of two inputs strings $(s, r)$ splits each normalized string into the set of all overlapping sub-strings of given lengths, termed \emph{$k$-shingles} (or \emph{$k$-grams}), where $k$ is the length of the sub-strings. We use small letters to denote strings or the corresponding records, and capital letters to denote their associated sets of $k$-shingles. The Jaccard index for records \aliceElement, \bobElement is
\begin{align}\label{eq:Jaccard}
J(\aliceElement, \bobElement) = \frac{\abs{S \cap R}}{\abs{S \cup R}}
\end{align}
when the context is clear we use $J$ instead of $J(\aliceElement, \bobElement)$.

\begin{example}
Consider the strings:
\begin{align*}
& \aliceElement = \texttt{`Sun\textbf{s}et Blvd, Los Angeles'} \\
& \bobElement = \texttt{`Sunet Blvd, Los Angeles'}
\end{align*}
that are normalized into 
\begin{align*}
& \texttt{`sun\textbf{s}et blvd los angeles'} \\
& \texttt{`sunet blvd los angeles'}
\end{align*}
and then split into the set of $19$ and $18$ shingles of length $k=5$, respectively: 
\begin{align*}
S =\{& \texttt{`sunse', `unset', `nset ', `set b', ..., `ngele', `geles'}\} \\
R =\{& \texttt{`sunet', `unet ', `net b', `et bl', ..., `ngele', `geles'}\} 
\end{align*}
Here, the Jaccard index is $J=\frac{15}{22} \approx 0.68$. Using longer shingles of length $k=11$ would result in a lower Jaccard index of $J=0.56$.
\end{example}

It is possible to instantiate a \gls{PPRL} solution that relies on the Jaccard index. The drawback of such a protocol is that it has quadratic complexity in the size of the datasets. For linear complexity, we use Min-Hash.

\begin{definition}[Min-Hash \cite{ullman}]
For a collision-resistant hash function $H$ with an integer output digest and an integer $k$, a Min-Hash function receives a string $s$ as input, converts it to a $k$-shingles set $S$, and returns 
\begin{align*}
    \texttt{MinH}^H_k(s) = \min_{e \in S}{H(e)}
\end{align*}
\end{definition}
When the context is clear we write MinH instead of MinH$^H_k$.

\begin{observation}[\hspace{1sp}\cite{ullman}]
\label{obs:mihashjac}
For two normalized records \aliceElement and \bobElement, a collision resistant hash function $H$, and $k>0$, it follows that $Pr\left[MinH^H_k(\aliceElement)=MinH^H_k(\bobElement)\right] = \jaccard{\aliceElement}{\bobElement}$.
\end{observation}

An \gls{LSH} involves applying $P$ different Min-Hash functions to a string $s$. The outputs are split into $B$ \emph{bands} of $R$ digests ($P = B R$). The concatenation of the $R$ digests of each band is again hashed to produce the \emph{signature} of the band, where the same signature hash function is used for all bands. An \gls{LSH} output is a tuple with these band signatures. 

\begin{definition}[LSH]
For $k,R,B \in \N$, $P=RB$, distinct collision-resistant hash functions $H_i$, $1 \le i \le P$ and another collision-resistant hash functions $G$, a band $b_j$, $1 \le j \le B$ over a string $s$ is the concatenation
\begin{align*}
    b_j(s) = \texttt{MinH}^{H_{R \cdot (j-1)+1}}_k(s)~|~\cdots~|~ \texttt{MinH}^{H_{R \cdot (j-1) + R}}_k(s)
\end{align*}
and the LSH output over a string $s$ is the ordered list
\begin{align*}
    LSH(s) = \left[
    G\left(b_1(s)\right),
    G\left(b_2(s)\right),
    \ldots,
    G\left(b_{B}(s)\right)
    \right]
\end{align*}
\end{definition}
 
Two \gls{LSH} tuples are considered to be a match if they share at least one common signature. We denote this by the indicator function
\begin{align*}
\texttt{LSHMatch} : \recordSpace \times \recordSpace & \longrightarrow \{0,1\} \\
(\aliceElement, \bobElement) & \longmapsto 
\begin{cases}
1 & 1 \le \sum\limits_{i=1}^{B}{Eq\left(LSH(\aliceElement)[i], LSH(\bobElement)[i]\right)} \\
0 & otherwise
\end{cases}
\end{align*}

\begin{observation}[\hspace{1sp}\cite{ullman}]
For two records \aliceElement, \bobElement,
\begin{align}\label{eq:P_J}
  Pr[\lshMatch{\aliceElement}{\bobElement}=1] = 1-(1-J^R)^B
\end{align}
\end{observation}

\begin{example}
Figure~\ref{fig:LSH} demonstrates an \gls{LSH} with $P=100$, $B=25$, $R=4$, where $\texttt{MinH}^{H_1}_5(s_1) = 17$, $\texttt{MinH}^{H_2}_5(s_1) = 43$, etc.
Subsequently, every sequence of $R$ digests is concatenated and hashed to produce a band signature, with a total of $B$ band signatures, which form the LSH of $s_1$, $LSH(s_1)=(865, 1082,\ldots, 172)$. Repeating the process for $s_2$, we observe a match in the signature of the second band for the two compared strings; this means that the two LSHs match and the strings match with a high probability.
\end{example}

\subsection{Private set intersection (PSI)}\label{sec:psi}
\gls{PSI} is a cryptographic protocol that allows two parties to compute the intersection of their private sets without revealing anything beyond this fact or  beyond the size of the intersected sets to the other party. \gls{PSI} is a special case of \gls{PPRL}, which considers only exact matches. Some variations of \gls{PSI} allow the parties to learn just the cardinality of the intersection. 

Many \gls{PSI} solutions exist (see Section \ref{sec:intro}). In this work, we use a unidirectional variant of the \gls{DH}-\gls{PSI} \cite{psi-orig}, as presented in Figure \ref{fig:dhpsi}. The two parties \party{\alice} and \party{\bob} first agree on a group $\G$ and a collision-resistant hash function $H$, and each party generates its own secret key $sk_\aliceElement$ and $sk_\bobElement$, respectively. Subsequently, both parties hash and encrypt their records using their private keys and send them to the other party. In addition, \party{\bob} encrypts the output of \party{\alice} using its secret key and sends the results back to \party{\alice}. Finally, \party{\alice} learns the intersection of the two datasets.

\begin{figure}[H]
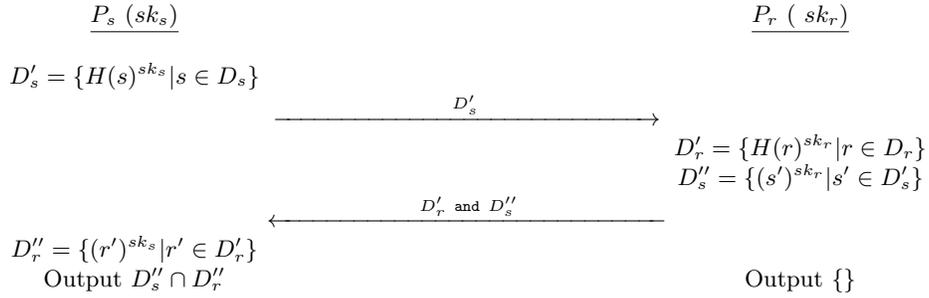

    \centering
    \begin{tabular}{ccc}
        \underline{\party{\alice}($sk_\aliceElement$)} & & \underline{\party{\bob}( $sk_\bobElement$)} \\ &&\\

        $\db{\alice}' = \{H(\aliceElement)^{sk_\aliceElement} |  \aliceElement \in \db{\alice}\}$ & & \\

        & $\xrightarrow{\hspace*{2.3cm} \db{\alice}' \hspace*{2.3cm}}$ & \\

        & & $\db{\bob}' = \{H(\bobElement)^{sk_\bobElement} |  \bobElement \in \db{\bob}\}$ \\
        & & $\db{\alice}'' = \{(\aliceElement')^{sk_\bobElement} | \aliceElement' \in \db{\alice}'\}$\\

        & $\xleftarrow{\hspace*{1.9cm} \db{\bob}' \texttt{ and } \db{\alice}'' \hspace*{1.9cm}}$ & \\

        $\db{\bob}'' = \{(\bobElement')^{sk_\aliceElement} | \bobElement' \in \db{\bob}'\}$ \\
        Output $\db{\alice}'' \cap \db{\bob}''$ & & Output $\{\}$
    \end{tabular}
    \caption{One side \gls{DH}-\gls{PSI}}
    \label{fig:dhpsi}
\end{figure}

Informally, the security of these protocols against semi-honest adversaries is guaranteed by the one-way property of the hash function, the computational hardness of the decisional \gls{DH}, and the one-more-\gls{DH} \cite{one-more} assumptions (see definitions in Appendix \ref{app:secdef}). The decisional \gls{DH} is used to hide the data in transit from eavesdroppers, while the one-more-\gls{DH} assumption is used to prevent  \party{\alice} from generating new records in the name of \party{\bob}.

One \gls{DH}-\gls{PSI} variant is the mutual \gls{DH}-\gls{PSI}, which includes one extra round: \party{\alice} sends $\db{\bob}''$ to \party{\bob} so that \party{\bob} can also compute the intersection. However, here an eavesdropper learns both $\db{\alice}''$ and $\db{\bob}''$ and can therefore learn the cardinality of the intersection $\db{\alice}'' \cap \db{\bob}''$. 

One issue with \gls{DH}-\gls{PSI} is that it is susceptible to man-in-the-middle attacks \cite{dh-psi-attack}. To mitigate this attack and the leakage of the mutual \gls{DH}-\gls{PSI}'s intersection cardinality,we assume that the transportation is encrypted and authenticated using TLS 1.3. 

\section{PPRL}\label{sec:pprl}

\gls{PPRL} is an \gls{ER} protocol between two parties \party{\alice} and \party{\bob}, with private datasets \db{\alice} and \db{\bob} of sizes \aliceNRecords and \bobNRecords, respectively; these records have a similarity measure \measure{\cdot}{\cdot}, and some additional privacy requirements. These requirements may lead to several security models and several formal definitions of \gls{PPRL}. 

The most intuitive way to define privacy for \gls{PPRL} is by following the \gls{PSI} privacy notion: \party{\alice} only learns $\bobNRecords$ and the intersection $\db{\alice} \cap \db{\bob}$, i.e., all records that exactly match in all fields while \party{\bob} only learns \aliceNRecords. Note that in both \gls{PSI} and \gls{PPRL}, \party{\alice} and \party{\bob} need to share the nature of the information contained in their datasets with each other to decide which \glspl{QID} they can validly compare. 

The difference between \gls{PSI} and \gls{PPRL} is that \gls{PSI} only returns exact matches according to some uniquely identifying \glspl{QID}, while \gls{PPRL} returns matching records up to some similarity indicator and according to non-unique \glspl{QID}. For example, a \gls{PSI} protocol may rely on users' \glspl{SSN}, while a \gls{PPRL} protocol may compare first and last names. Thus, a \gls{PPRL} may inadvertently match ``David Doe'' with ``Davy Don'' even if they represent different entities (users).

\begin{figure}[ht!]
    \centering
    \includegraphics[width=0.7\linewidth]{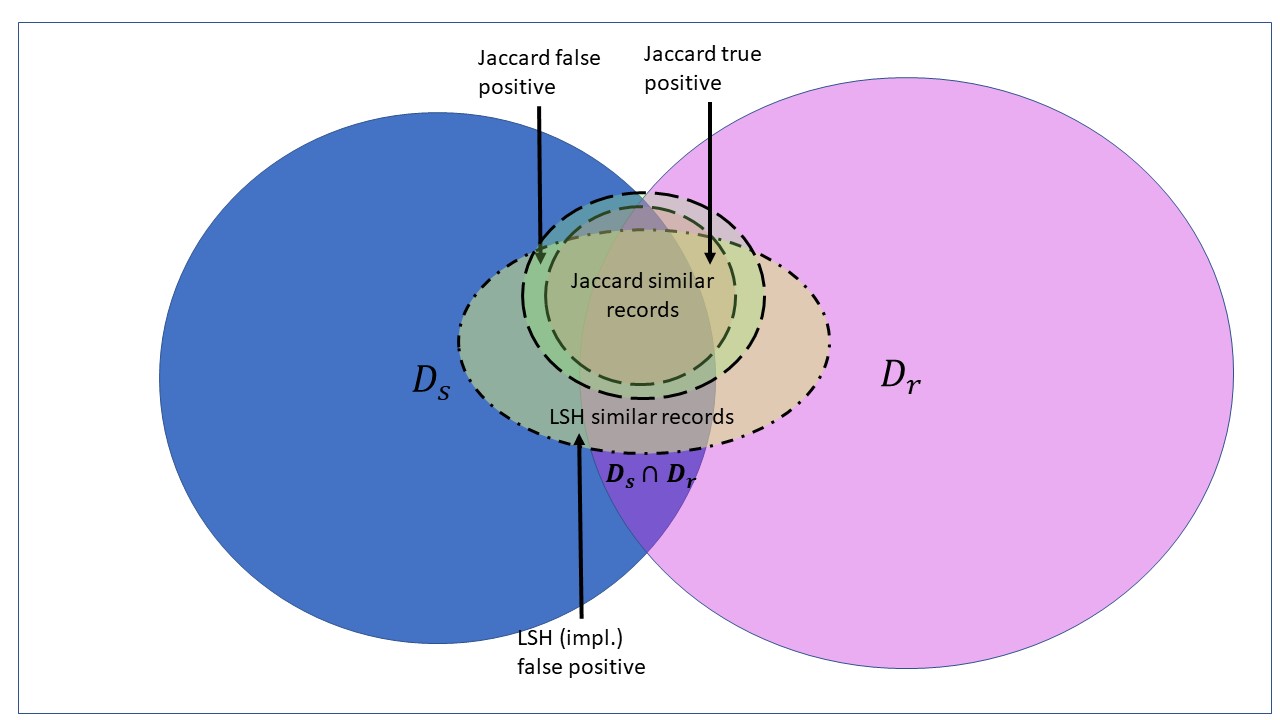}
    \caption{A Venn diagram of different \gls{ER} outputs applied on two datasets \db{\alice} and \db{\bob}. The \gls{ER} methods are: matching only identical pairs of records (purple), matching pairs of records with a Jaccard index above some threshold (green), and matching pairs of records with matching \gls{LSH} indicators (yellow).}
    \label{fig:venn}
\end{figure}

Figure \ref{fig:venn} shows a Venn diagram for the output of different \gls{ER} solutions on \db{\alice} and \db{\bob} datasets. With the exact matching method ($\db{\alice} \cap \db{\bob}$) no privacy risks occur since it only reveals the agreed-upon intersection\footnote{In practice, if \party{\alice} learns that both parties share a record with the same \gls{SSN} and at a later stage learns that the other record fields do not match, then it may deduce that \db{\bob} contains a record with a very close \gls{SSN} that leaks information. Following previous studies, we only consider leaks that occur as a result of the protocol itself.}. In contrast, when using the Jaccard similarity to compute the matches, the parties learn: a) records in $\db{\alice} \cap \db{\bob}$, which is ok; b) records outside $\db{\alice} \cap \db{\bob}$ that represent the same entity (true-positive), which is also ok; c) records outside $\db{\alice} \cap \db{\bob}$ that represent different entities (false-positive), which may break the privacy of the parties. In general, any \gls{PPRL} protocol must assume this kind of leakage, and should do its best to quantify it, e.g., by assuming the existence of a bound $\tau$ on the similarity false-positive rate.

\begin{definition}[\gls{PPRL}]
\label{def:pprl}
A \gls{PPRL} protocol $\mathcal{P}$ between two parties $\party{\alice}$, $\party{\bob}$ with datasets $\db{\alice}$, $\db{\bob}$, respectively, a similarity measure \measureE, a measure indicator \indicatorE{\measureE}{t} for a fixed threshold $t$ with a false-positive rate bounded by $\tau$, has the following properties.
\begin{itemize}
    \item \textbf{Correctness:} $\mathcal{P}$ is correct if it outputs to \party{\alice} the set
    \[
    res = \{ (\aliceElement, Enc(\bobElement)) ~|~ \aliceElement \in \db{\alice}, \bobElement \in \db{\bob}, \indicator{\measureE}{t}{\aliceElement}{\bobElement}=1\},
    \] where $Enc(\bobElement)$ is an encryption of $\bobElement$ under a secret key of \party{\bob}.
    \item \textbf{Privacy:} $\mathcal{P}$ maintains privacy if \party{\alice} only learns $res$ and $\bobNRecords$, and \party{\bob} only learns $\aliceNRecords$.
\end{itemize}
\end{definition}

\begin{corollary}\label{cor:priv}
The leaked information of \party{\bob} in $\mathcal{P}$ is bounded by $\tau \cdot \frac{|res|}{\bobNRecords}$.
\end{corollary}

Definition \ref{def:pprl} assumes the existence of $\tau$ but only implicitly uses it. The reason is that $\tau$ does not always exist. In many cases, it can be empirically estimated based on prior data or based on perturbed synthetic data. However, relying solely on empirical estimates increases the ambiguity of the privacy definition for such protocols. Moreover, in many cases, $\tau$ depends on data from the two datasets that have different distributions, which none of the parties know in advance. Another reason for only implicitly relying on $\tau$ is that the leaked information in Cor. \ref{cor:priv} depends on $res$ and can only be computed after running the protocol. 

While $\tau$ bounds the privacy leak from above, there is still the issue of quantifying the exact leakage after the protocol ends. It is not clear how the parties can verify the number of false-positive cases without revealing private data. Usually, an \gls{ER} protocol is used when the compared records do not include uniquely identifying fields (such as an \gls{SSN}) and thus the parties cannot compute the exact matches using PSI. Consequently, their only way to verify matches is by revealing their private data. To assist in this task, we define a protocol called a revealing \gls{PPRL}.

\begin{definition}[Revealing \gls{PPRL}]
\label{def:rpprl}
A revealing \gls{PPRL} protocol $\mathcal{P}$ is a \gls{PPRL} protocol $\mathcal{P}'$, where \party{\bob} also learns
$
u = \{Enc(\bobElement) ~|~ (\aliceElement, Enc(\bobElement)) \in \mathcal{P}'.res\}
$
and \party{\alice} also learns
\[
    res' = \{ (\bobElement, Enc(\bobElement)) ~|~ (\aliceElement, Enc(\bobElement)) \in \mathcal{P}'.res\},
\]
\end{definition}

In words, \party{\bob} learns which of its own records are matched, and \party{\alice} learns the field content of the matched records of the other party. 
The simplest way to achieve a revealing \gls{PPRL} is for \party{\alice} to send $u$ to \party{\bob}, who will then decrypt its values and hand them back to \party{\alice}. The difference between Definitions \ref{def:pprl} and \ref{def:rpprl} is that in the latter, \party{\alice} learns the values of \party{\bob}'s records instead of just their encryption. While this definition leaks more data from \party{\bob} to \party{\alice}, it is easier to analyze because now \party{\alice} can verify the matches with some probability and learn the estimated number of false-positives. 
We also consider the definitions of the associated mutual \gls{PPRL} and the mutual revealing \gls{PPRL}.

\begin{definition}[Mutual \gls{PPRL}]
\label{def:mpprl}
A \gls{PPRL} protocol $\mathcal{P}$ between two parties $\party{\alice}$, $\party{\bob}$ with datasets $\db{\alice}$, $\db{\bob}$, respectively, a similarity measure \measureE, a measure indicator \indicatorE{\measureE}{t} for a fixed threshold $t$ with a false-positive rate bounded by $\tau$, has the following properties.
\begin{itemize}
    \item \textbf{Correctness:} $\mathcal{P}$ is correct if it outputs $res_\alice$ (resp. $res_\bob$) to \party{\alice} (resp. \party{\bob}), where
    \[
    res_\alice = \{ (\aliceElement, Enc(\bobElement)) ~|~ \aliceElement \in \db{\alice}, \bobElement \in \db{\bob}, \indicator{\measureE}{t}{\aliceElement}{\bobElement}=1\}
    \]
    \[
    res_\bob = \{ (\bobElement, Enc(\aliceElement)) ~|~ \aliceElement \in \db{\alice}, \bobElement \in \db{\bob}, \indicator{\measureE}{t}{\aliceElement}{\bobElement}=1\},
    \]
    and $Enc(\bobElement)$ (resp. $Enc(\aliceElement)$) is an encryption of $\bobElement$ (resp. \aliceElement) under a secret key of \party{\bob} (resp. \party{\alice}).
    \item \textbf{Privacy:} $\mathcal{P}$ maintains privacy if \party{\alice} only learns $res_\alice$ and $\bobNRecords$, and \party{\bob} only learns $res_\bob$ and $\aliceNRecords$.
\end{itemize}
\end{definition}

The mutual revealing \gls{PPRL} is similarly defined. The difference between the mutual \gls{PPRL} and the revealing \gls{PPRL} in terms of privacy is that in the mutual \gls{PPRL}, \party{\bob} can match the encryption of \party{\alice} records to its records and therefore gains more information while \party{\alice} only learns the encryption of \party{\bob} records.

In the \gls{PPRL} protocols described above, the two parties learn the intersection of their datasets. However, in some scenarios, the parties merely need to learn the \emph{number} of matches and do not wish to reveal the identity of the matched records to the other party. To this end, we define an N-\gls{PPRL} protocol.

\begin{definition}[N-\gls{PPRL}]
\label{def:npprl}
A \gls{PPRL} protocol $\mathcal{P}$ between two parties $\party{\alice}$, $\party{\bob}$ with datasets $\db{\alice}$, $\db{\bob}$, respectively, a similarity measure \measureE, a measure indicator \indicatorE{\measureE}{t} for a fixed threshold $t$ with a false-positive rate bounded by $\tau$, has the following properties.
\begin{itemize}
    \item \textbf{Correctness:} $\mathcal{P}$ is correct if it outputs to \party{\alice} the value
    \[
    N_{\aliceElement \cap \bobElement} = |\{ (\aliceElement, \bobElement) ~|~ \aliceElement \in \db{\alice}, \bobElement \in \db{\bob}, \indicator{\measureE}{t}{\aliceElement}{\bobElement}=1\}|,
    \]
    \item \textbf{Privacy:} $\mathcal{P}$ maintains privacy if \party{\alice}, (resp. \party{\bob}) only learns $N_{\aliceElement \cap \bobElement},\bobNRecords$ (resp. $\aliceNRecords$).
\end{itemize}
\end{definition}

The mutual N-\gls{PPRL} protocol is similarly defined.

\section{Our solution}\label{sec:solution}

Our \gls{PPRL} solution (hereafter: LSH-PSI PPRL) is an \gls{ER} protocol that uses \lshMatchE as its similarity indicator, where for privacy reasons, the parties cannot directly share the \gls{LSH} results. The reason depends on whether the \gls{LSH} is a preimage-resistant hash function or not. When it is not, \party{\bob} and \party{\alice} can simply inverse the \gls{LSH} results for records that are not in the intersection and reveal private information of \party{\alice}, \party{\bob}, respectively. But even when it is, the solution's privacy depends on the \gls{LSH} input entropy, where the parties can maintain an offline brute force attack against the \gls{LSH} records of the other party.

\begin{figure}[ht!]
    \centering
    \includegraphics[width=0.7\linewidth]{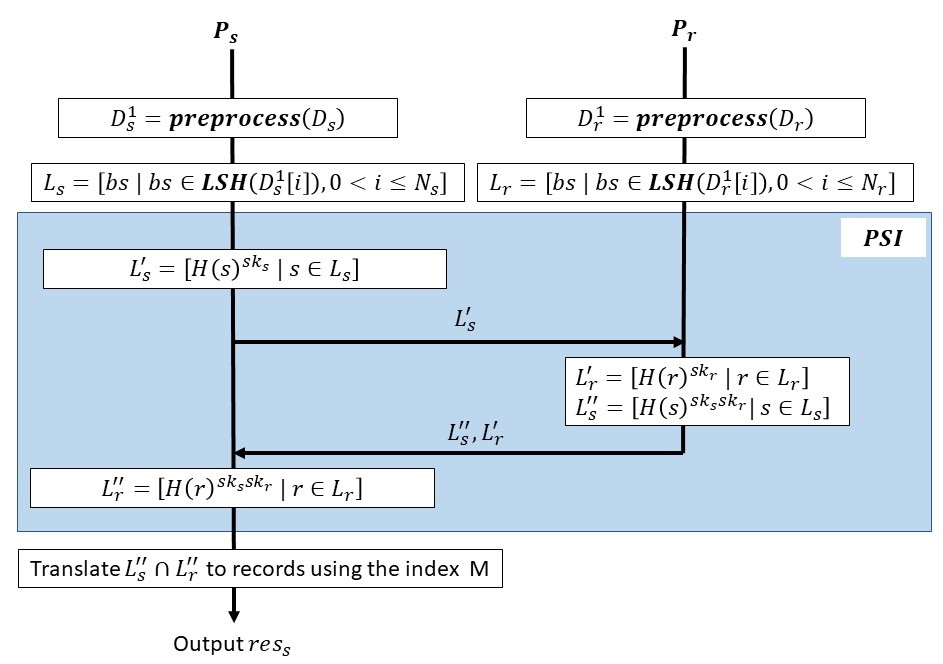}
    \caption{Schematic of the LSH-PSI \gls{PPRL} protocol.}
    \label{fig:illustration2}
\end{figure}

To mitigate the privacy issue, we use a \gls{PSI} protocol. The two parties first compute the \gls{LSH} band signatures of all their records and then apply a \gls{PSI} protocol over these signatures. Finally, \party{\alice} maps back the intersected signatures to the original records to learn the set of similar records. The concrete properties of \lshMatchE can be tuned using the $B$ and $R$ \gls{LSH} parameters. Figure \ref{fig:pprlproto} presents the LSH-PSI protocol, and Figure \ref{fig:illustration2} illustrates it schematically.

Our protocol is defined against semi-honest (honest-but-curious) adversaries, where all parties do not deviate from the protocol, and their inputs are genuine. Nevertheless, they may record and analyze all the intermediate computations and messages from the other parties to get more information. 

\begin{figure}[ht!]
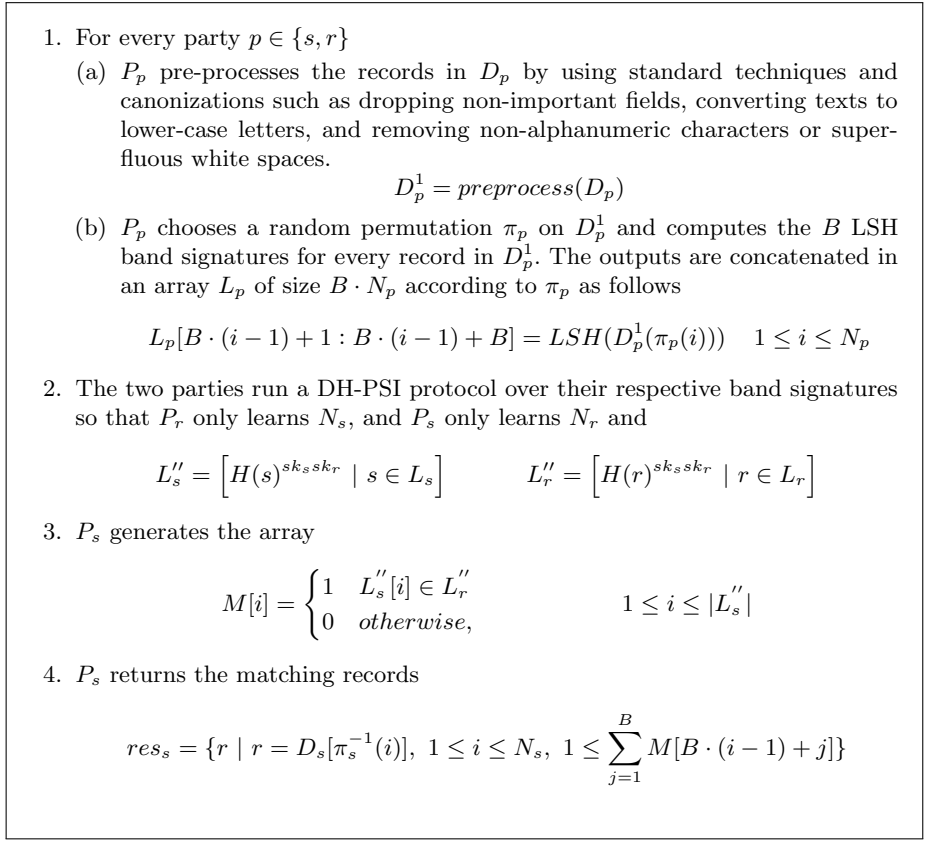

    \centering
\begin{framed}
\begin{enumerate}
\item For every party $p \in \{\alice, \bob\}$

\begin{enumerate}
\item $\party{p}$ pre-processes the records in $\db{p}$ by using standard techniques and canonizations such as dropping non-important fields, converting texts to lower-case letters, and removing non-alphanumeric characters or superfluous white spaces.
\[
\db{p}^1 = preprocess(\db{p})
\]

\item \party{p} chooses a random permutation $\pi_{p}$ on $\db{p}^1$ and computes the $B$ \gls{LSH} band signatures for every record in $\db{p}^1$. The outputs are concatenated in an array $L_p$ of size $B \cdot N_p$ according to $\pi_{p}$ as follows
\begin{align*}
L_p[B\cdot (i-1) + 1: B \cdot (i-1) + B]  = LSH(\db{p}^1(\pi_{p}(i))) 
& & 1 \le i \le N_p
\end{align*}

\end{enumerate}

\item The two parties run a \gls{DH}-\gls{PSI} protocol over their respective band signatures so that \party{\bob} only learns \aliceNRecords, and \party{\alice} only learns \bobNRecords and
\begin{align*}
    L''_\alice = \left[H(\aliceElement)^{sk_\alice sk_\bob} ~ | ~ \aliceElement \in L_\alice\right]    & &
    L''_\bob = \left[H(\bobElement)^{sk_\alice sk_\bob} ~ | ~ \bobElement \in L_\bob\right]
\end{align*}

\item \party{\alice} generates the array 
\begin{align*}
M[i] = \begin{cases}
1 & L^{''}_\alice[i] \in L^{''}_\bob \\
0 & otherwise,
\end{cases}
& & 1 \le i \le |L^{''}_\alice|
\end{align*}

\item \party{\alice} returns the matching records
\[
res_\alice = \{r ~|~ r = \db{\alice}[\pi_\alice^{-1}(i)],~
                    1 \le i \le \aliceNRecords,~
                    1 \le \sum_{j=1}^{B}{M[B\cdot (i-1) + j]} \}
\]
\end{enumerate}
\end{framed}
    \caption{The LSH-PSI PPRL protocol.}
    \label{fig:pprlproto}
\end{figure}

\begin{remark} 
The two parties must use the same preprocessing techniques to increase the efficiency of the underlying \gls{ER} method. In addition, the LSH-PSI protocol assumes that the pre-processing phase runs some deduplication protocol on the dataset of every party. Otherwise, \party{\alice} can extract information from pairs of matching records $\aliceElement_1, \aliceElement_2 \in \db{\alice}$, where $\aliceElement_1$ matches a record in \db{\bob} but $\aliceElement_2$ does not.
\end{remark} 

\begin{remark}
The \gls{PSI} protocol is executed for all records at once and not per record, therefore it is critical to preserve the order of the signatures exchanged between the parties, i.e., of $L^{'}_\alice$ and $L^{''}_\alice$ in Figure \ref{fig:illustration2}. Otherwise, it will be impossible to match the records in Step 4 of Figure \ref{fig:pprlproto}. In Section~\ref{sec:card}, we discuss the case where \party{\bob} does \emph{not} preserve the order of \party{\alice} encrypted signatures.
\end{remark}

The purpose of using the permutation $\pi_p$ in Step 1b is to avoid the case where the other party learns information about ``missing" records. For example, suppose that the records in \db{\bob} are ordered alphabetically according to a first name \gls{QID}, and that \party{\alice} learns that $Jerry$ and $Joseph$ are in the intersection. If $Jerry$ and $Joseph$ happen to belong to adjacent records in \db{\bob}, then an honest but curious \party{\alice} learns that \party{\bob} has no record for $John$. When using a permutation, the only way for \party{\alice} to deduce the same information is by learning all the records in \db{\bob}. A concrete example of Steps 1.b - 3 is given in Appendix \ref{sec:example}.

\begin{theorem}\label{th:pprl}
The LSH-PSI PPRL protocol is a PPRL protocol according to Definition \ref{def:pprl} where the similarity indicator is \lshMatchE. This protocol is secure against semi-honest adversaries.
\end{theorem}

\begin{proof}

\textbf{Correctness.} The correctness of the protocol follows from the fact that the intersection $L_{\alice} \cap L_{\bob}$ has a one-to-one correlation with the encrypted band signatures $L''_{\alice} \cap L''_{\bob}$.

\textbf{Privacy of \party{\alice}.} By the discrete-log assumption, \party{\bob} only gets to see $\aliceNRecords$ elements that are indistinguishable from random values. Thus, \party{\bob} only learns $\aliceNRecords$.

\textbf{Privacy of \party{\bob}.}
\party{\alice} gets from \party{\bob} the values of $L'_{\alice}$ and $L_{\bob}$ raised to the power of \party{\bob}'s secret key. By the discrete-log assumption, these values are indistinguishable from random to \party{\alice}. Except that \party{\alice} can raise $L'_{\bob}$ values to the power of its own secret key and then intersect the results with $L''_{\alice}$. This intersection of random values is used by \party{\alice} to identify matching signatures, which is expected by Definition \ref{def:pprl}. Because \party{\alice} learns nothing from values outside the intersection, we say that it only learns $res$ and $\bobNRecords$ as expected.
\qed
\end{proof}

\begin{remark}
Similar to the \gls{DH}-\gls{PSI} case, the use of TLS 1.3 allows the parties to mutually authenticate themselves and to avoid the attack presented in \cite{dh-psi-attack}. Still, as a defense-in-depth mechanism, the parties in every \gls{PPRL} session should avoid reusing secret keys to avoid man-in-the-middle attacks.
\end{remark}

\subsection{PPRL variants}\label{sec:card}

Based on the above protocol, we construct three other protocols: a mutual \gls{PPRL} protocol, an N-\gls{PPRL} protocol, and a revealing \gls{PPRL} protocol, where the latter immediately follows the definition.

\paragraph{A mutual \gls{PPRL} protocol.} To establish a mutual \gls{PPRL} protocol, we modify Step 2 of Figure \ref{fig:pprlproto} to use the mutual \gls{DH}-\gls{PSI} protocol of Section \ref{sec:psi}. The security of the protocol follows from either the security of the mutual \gls{DH}-\gls{PSI}, or from the fact that the mutual protocol is equivalent to running the original \gls{PPRL} protocol twice: first between \party{\alice} and \party{\bob}, and subsequently between \party{\bob} and \party{\alice}. Note that \party{\alice} cannot reduce the communication by sending only records that are in the intersection because then an eavesdropper can learn the intersection size. This claim is valid even when using a secure communication channel (e.g., TLS 1.3).

\paragraph{An N-\gls{PPRL} protocol.} To achieve an N-\gls{PPRL} protocol, we could have simply counted the number of elements in the intersection set $res$, but this would reveal to \party{\alice} more information beyond $N_{\aliceElement \cap \bobElement}$. Instead, we suggest reordering the encrypted band signatures during the \gls{DH}-\gls{PSI} in a way that hides the identity of the matched records but still enables them to be counted. Specifically, we ask \party{\bob} to apply a secret permutation to $L_{\aliceElement}''$ before sending it to \party{\alice}. This permutation has a special property that permutes together the groups of adjacent $B$ signatures that originate from the same record, otherwise, \party{\alice} will not be able to distinguish between the cases
\begin{enumerate}
    \item $|LSH(\aliceElement_1) \cap LSH(\bobElement_1)| = 1$ and $|LSH(\aliceElement_2) \cap LSH(\bobElement_2)| = 1$
    \item $|LSH(\aliceElement_1) \cap LSH(\bobElement_1)| = 2$ and $|LSH(\aliceElement_2) \cap LSH(\bobElement_2)| = 0$
\end{enumerate}
We call the above permutation an intra-permutation of records. In addition, we apply an inter-permutation of records, where we separately permute the $B$ signatures in each group of signatures in $L''_{\aliceElement}$ that originate from the same record.

\section{Experiments}\label{sec:expr}

\subsubsection{Experimental setup.}
We carried out the experiments on two machines that are located in different \glspl{LAN}. We measured an average of $65$ ms round-trip latency between them.
\begin{itemize}
\item Machine $A$ has an Intel\circler~Xeon\circler~CPU E5-2620 v3 @ 2.40GHz, with 12 physical cores and 377 GB of RAM. 
\item Machine $B$ has an Intel\circler~Xeon\circler~CPU E5-2699 v4 @ 2.20GHz, with 44 physical cores and 744 GB of RAM. 
\end{itemize}
We set machine $A$ to run \party{\alice} and machine $B$ to run \party{\bob} with $\aliceNRecords \approx \bobNRecords$.

Our code is written in C++ and runs on Ubuntu 20.04. It uses OpenSSL version 1.1.1f to establish secure TLS 1.3 connections between the two parties. In addition, it uses OpenSSL hash function implementation (concretely, H=\pcr{SHA256}) and \gls{DH} operations (concretely, elliptic curve \gls{DH} operations over the NIST P-256 curve). We report communications in KB and running time in seconds. We also provide a breakdown of the different running time phases: communication and computations per party. For the measurements, we separated the communication phases from the computation phases, which in a real scenario can be pipelined to run in parallel.

For the evaluations, we considered two dataset cases: a) The North Carolina voter register (NCVR) dataset\footnote{\url{https://www.ncsbe.gov/results-data/voter-registration-data, last accessed Mar 2022.}}, which is commonly used for PPRL evaluations; b) a synthetic dataset that we generated and made available in \cite{helayers}.

\paragraph{NCVR datasets.} We used the November 2014 and November 2017 snapshots of the NCVR datasets. Prior to running the PPRL protocol, we deduplicated the snapshots by eliminating duplicate records with identical ``NCID'' or with identical values in the `first\_name', `last\_name', `midl\_name', `birth\_place' and `age' fields. Subsequently, we removed the NCID field from the two snapshots, and ran our PPRL protocol on the two snapshots. A reported matching pair was considered to be a true-positive event if the two reported records share the same NCID value. Table \ref{tab:ncvr} shows the accuracy breakdown of the LSH we used by reporting the number of false-negative (FN), false-positive (FP), and true-positive (TP) events, together with the precision, recall, and F1 results when sampling sets of fixed sizes from the above snapshots. Note that while the precision is high, the absolute number of false-positives may be regarded as too high for some users. See Section~\ref{subsec:optBR} for ways to tune the process and balance the number of false positive and false negative cases while considering the protocol performance. 

\begin{table}[t!]
    \centering
    \setlength{\tabcolsep}{8pt}
    \caption{Accuracy of our PPRL protocol over the NCVR snapshots.}
    \label{tab:ncvr}
    \begin{tabular}{crrrrrr }
        \hline
         Set size & FN & FP & TP & Precision (\%) & Recall (\%) & F1 (\%) \\
         \hline
         $10^4$ & 19 & 21 & 653 & 96.88 & 97.17 & 97.03 \\
        $10^5$ & 1,369 & 1,665 & 55,682 & 97.1 & 97.6 & 97.35 \\
        $10^6$ & 22,233 & 19,365 & 847,724& 97.77 & 97.44 & 97.61\\
        \hline
    \end{tabular}
\end{table}

\paragraph{Synthetic dataset.} We generated two synthetic datasets using IBM InfoSphere\circler{} Optim$^{TM}$ Test Data Fabrication \cite{data-fabrication} with the following fields: `first name',  `last name', `email', `email domain', `address number', `address location', `address line', `city', `state', `country', `zip base', `zip ext', `phone area code', `phone exchange code' and `phone line number', where $\aliceNRecords \approx \bobNRecords \approx 1,000,000$. We generated the datasets in a way that only $100$ records in the two datasets represent identical entities. The pairs of records that describe these shared entities sometimes have identical fields and sometimes fields with minor typos, different styles, and other types of minor differences, which are still small enough to warrant the assumption that the similar records in fact describe the same entity. Our \gls{PPRL} protocol identified all the matching records. The performance evaluation of the protocol is given in Table \ref{tbl:synthetic}.

\begin{table}[t!]
\centering
\caption{Performance results on the synthetic dataset for different samples of the original dataset.}
\label{tbl:synthetic}
\begin{tabular}{ccccc}
    \textbf{Set sizes} & \textbf{Comm. (KB)} & \textbf{Comm. time (s)} & \textbf{Offline time (s)} &  \textbf{Total time (s)}         \\
    \hline
    $2^8$ & $5.68\cdot10^{3}$ & 3 & 1 & 4       \\
    $2^{12}$ & $9.04\cdot10^{4}$ & 17 & 2 & 19       \\
    $2^{16}$ & $1.44\cdot10^{6}$ & 237 & 36 & 273       \\
    $2^{20}$ & $1.19\cdot10^{7}$ & 1,959 & 608 & 2,567       \\
    \hline
\end{tabular}
\end{table}

\section{Conclusion}\label{sec:conc}
We presented a novel PPRL solution that relies on LSH to identify similar records while using PSI to ensure privacy. We formally defined the privacy guarantees that such a protocol provides and evaluated its efficiency. Our results show that it takes $11-45$ minutes (depending on the network settings) to perform a PPRL solution comparing two large datasets with $2^{20}$ records per dataset. Note that none of the results presented in Section~\ref{sec:rel} reported comparable speeds for such large datasets. This makes our solution practical and attractive for companies and organizations. We made our implementation available for testing at \cite{helayers}.

We proposed a PPRL framework that can use different PSI protocols as long as they provide the same security guarantees defined above. We demonstrated our solution using an ECDH PSI protocol. It may be an interesting direction to implement and test the protocol using other solutions that can further improve its performance and overall bandwidth.

\bibliography{main.bib} 
\bibliographystyle{splncs04.bst}

\appendix

\section{Security assumptions}\label{app:secdef}

\begin{definition}[Decisional \gls{DH} (DDH)]
For a cyclic group $\G$, a generator $g$, and integers $a,b,c \in \Z$, the decisional \gls{DH} problem is hard, if for every \gls{PPT} adversary $\A$
\begin{align*}
|Pr[A(g, & g^a, g^b, g^{ab}] = 1) - \\
         & Pr[A(g, g^a, g^b, g^c) = 1]| < negl(),    
\end{align*}
where the probability is taken over $(g, a, b, c)$.
\end{definition}

\begin{definition}[Computational \gls{DH} (CDH)]
For a cyclic group $\G$, a generator $g$, and integers $a,b \in \Z$, the computational \gls{DH} problem is hard, if for every \gls{PPT} adversary $\A$
\[
Pr[A(g, g^a, g^b] = g^{ab}) < negl(),
\]
where the probability is taken over $(g, a, b)$.
\end{definition}

\begin{definition}[One-more-\gls{DH} (OMDH) \cite{one-more}]
Let $\G$ be a cyclic group. The one-more-\gls{DH} problem is hard, if for every \gls{PPT} adversary $\A$ that gets a generator $g \in \G$ together with some power $g^a$ and who has access to two oracles: $h^a = CDH_{g,g^a}(h)$ for some $h \in \G$, and $r \xleftarrow{\$} C()$ a challenge oracle that returns
a random challenge point $r \in G$ and can only be invoked after all calls to the $CDH_{g, g^a}$, it follows that 
\[
Pr[A(g, g^a, r \leftarrow C()) = r^a] < negl()
\]
where the probability is taken over $(g, a)$.
\end{definition}


\section{Example of the LSH-PSI protocol}\label{sec:example}

A concrete example of Steps 1.b - 3 of the LSH-PSI PPRL protocol (Figure \ref{fig:pprlproto}) is given in Figure \ref{fig:bands1}. Suppose that $v = H(455)^{sk_\alice sk_\bob}$ then \party{\alice} learns via the \gls{PSI} process that \party{\bob} also has a band signature with the same value $455$. \party{\bob} took care to preserve the order of \party{\alice}'s encrypted band signatures during the \gls{PSI}, so \party{\alice} can map the shared value $v$ back to the band signature for Band $1$ of record $N_\alice$, and deduce that \party{\bob} has some unknown record that is similar to her own record $N_\alice$.

\begin{figure}
    \centering
    \includegraphics[width=\linewidth]{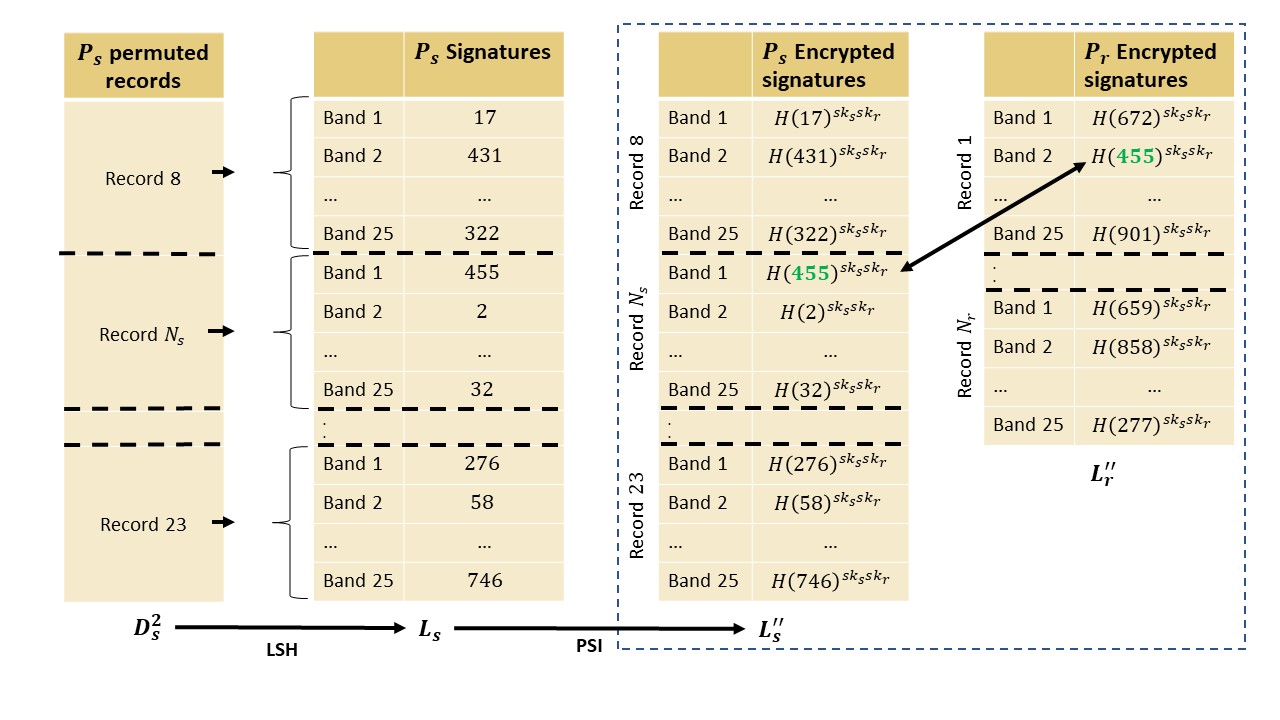}
    \caption{Steps 1.b - 3 of our protocol. \party{\alice} learns via the \gls{PSI} protocol that the signature for Record \aliceNRecords Band $1$ is shared with \party{\bob}.}
    \label{fig:bands1}
\end{figure}


\section{Using the Jaccard indicator}\label{sec:jaccard}

Theorem \ref{th:pprl} shows that the LSH-PSI PPRL protocol follows Definition \ref{def:pprl} when considering the LSH as the similarity indicator. This means that security reviewers need to accept the privacy leakage that occurs when using an \gls{LSH}, something that is already done by many organizations that perform RL. However, some reviewers may instead prefer to trust the Jaccard index due to its wide acceptance.

Figure \ref{fig:venn} shows two ways to define \gls{LSH} false-positive events: in relation to exact matches of entire records as in the LSH-PSI PPRL, or in relation to the method of matching pairs of records with a high enough Jaccard index. Thus according to the latter definition an LSH false-positive happens only when a pair of records are matched due to having at least one shared LSH band, and yet they do not have a high enough Jaccard index to justify a claim of similarity. Bounding the false-positive events rate  $\tau'$ based on the latter definition will allow us to define an LSH-PSI PPRL related to the Jaccard index metric but with a different bound $\tau \cdot \tau'$, where $\tau$ is the Jaccard original false-positive bound. In this section, we further discuss the relation between the LSH and the Jaccard index.

For two records $\aliceElement, \bobElement$ with Jaccard index $J$, Figure \ref{fig:P_J_a} shows the probability for an $\lshMatchE=1$ event according to Equation \ref{eq:P_J} with $R=200$ and $B=20$. In standard \gls{ER} solutions, it is the role of the domain expert to decide the specific Jaccard index that would indicate enough similarity between the two records. For example, in the figure the targeted Jaccard index is 0.78. The figure shows the cumulative probability of getting true-positives ($J(\aliceElement, \bobElement) > 0.78$ and $\lshMatch{\aliceElement}{\bobElement}=1$), true-negatives ($J(\aliceElement, \bobElement) \le 0.78$ and $\lshMatch{\aliceElement}{\bobElement}=0$), and the corresponding false-positive and false-negative cumulative probabilities. 

\begin{figure}[ht!]
    \centering
    \includegraphics[width=0.9\linewidth]{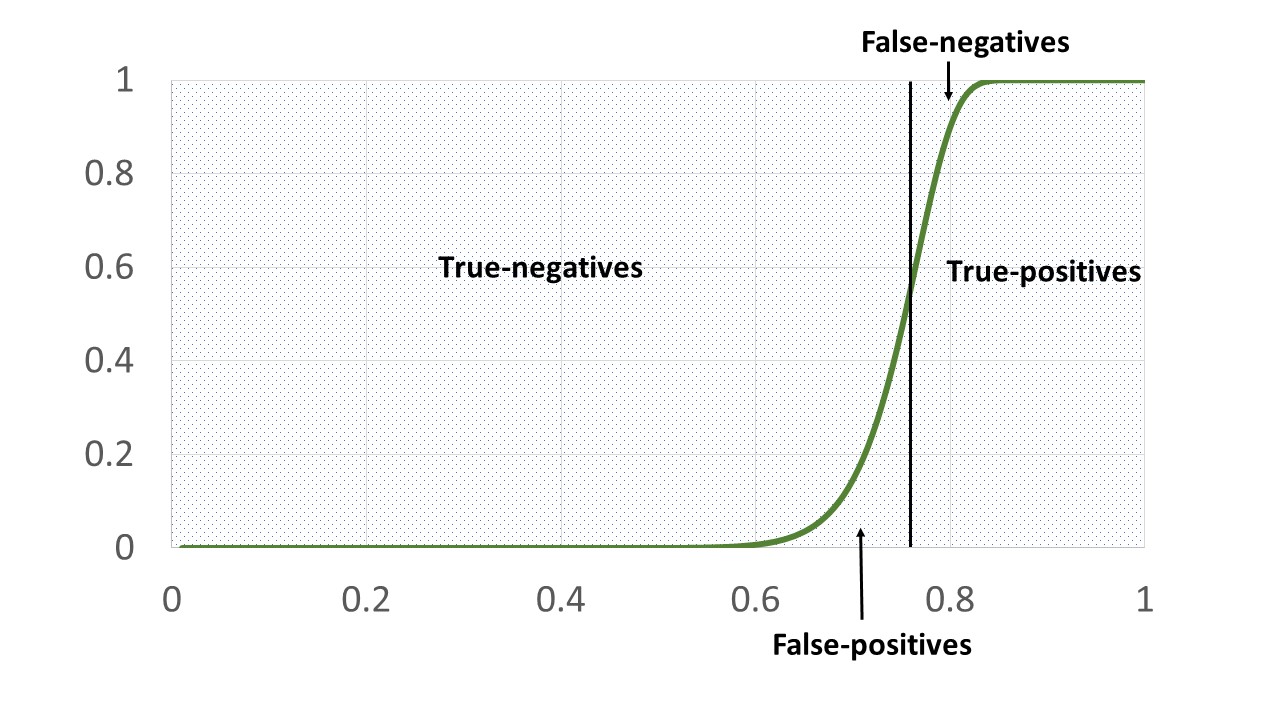}
    \caption{The function $F(J) = 1-(1-J^R)^B$ from Eq. \ref{eq:P_J}, where $R=20$ and $B=200$. The black vertical line is the Jaccard index threshold.}
    \label{fig:P_J_a}
\end{figure}

The above example shows that when $B=20$ and $R=200$, it is possible to close the gap between the Jaccard index and the LSH by choosing the Jaccard threshold to be below $0.5$. In that case, the probability for a false-positive event is less than $0.0001$, which means that one in every ten-thousand records leaks. However, using such a Jaccard threshold will yield many false-positive cases relative to exact record matching, which is less desirable in terms of privacy. 

It turns out that it is possible to tune the slope of the accumulated probability function. Figure~\ref{fig:P_J_b} compares the probability functions in four different setups $B=20, R=200$ (setup 1) $B=100, R=100$ (setup 2) $B=14, R=30$ (setup 3) and $B=120, R=18$ (setup 4). Here, we see that replacing setup 1 with setup 2 allows us to set the Jaccard threshold at 0.78 while reducing the LSH false-positive rate to as low as $10^{-8}$. However, setup 2 dramatically increases the LSH false-negative rate. Note however that false negatives affect the security less than false-positives, and in addition, users are often much more reluctant to report false positives than to miss reports due to false negatives. Setup 2 may also increase the overall performance of the protocol relative to setup 1 because there are many more bands to encrypt and communicate, as described in the following section.

\subsection{Optimizing the protocol}
\label{subsec:optBR}

Setup 4 in Figure~\ref{fig:P_J_b} probably results in more false-positive and false-negative cases than setup 1, and the low slope of the curve implies a larger region of uncertainty. However, the \gls{PSI} for setup 4 runs more than $6$ times faster than the \gls{PSI} for setup 1, because there are just $20$ rather than $180$ band signatures that need to be encrypted and communicated. The change in the $R$ parameter does not affect the performance as much, since it merely determines the number of Min-Hashes that need to be computed locally. It turns out that computing a Min-Hash (like the highly optimized SHA-256 operation) is much faster than computing the power in the the underlying groups of the \gls{DH} protocol. Moreover, there are known methods for quickly producing $R$ different permutations out of a single SHA-256 call such as the Mersenne twister~\cite{carter1979universal}. Finally, the value of $R$ does not affect the size of the communication.

\begin{figure}
    \centering
    \includegraphics[width=0.8\linewidth]{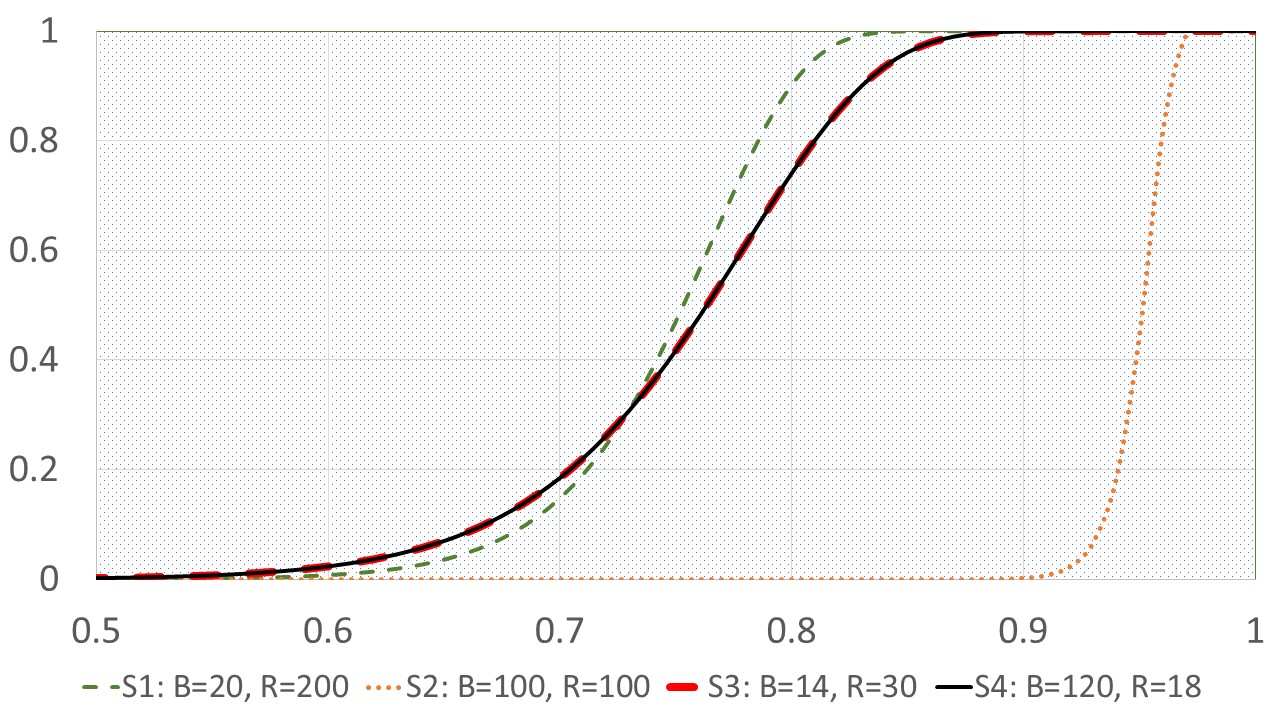}
    \caption{A comparison of four probability functions $F(J) = 1-(1-J^R)^B$ (see Eq. \ref{eq:P_J}) with different $B$ and $R$ values.}
    \label{fig:P_J_b}
\end{figure}

We use the $B$ and $R$ parameters to control the curve, which in turn affects the protocol's accuracy and performance. Reducing $B$ makes it less likely to find a matching band signature, thus increasing the false-negative probability, but improving performance. The rate of false-negatives can be reduced by decreasing $R$, thus making it more probable for two bands to match. Conversely, if the false-positive rate is too high, then one can increase $R$ with little performance penalty. We therefore optimize the process by searching for values of $B$ and $R$ that have the minimal $B$ value (for best performance) while more or less preserving the targeted curve shape. 

Suppose for example that setup 1 has the targeted probability function. The figure shows the probability function for setup 3, which runs almost twice as fast as setup 1 and has an almost identical probability function. Setup 4 has an almost identical curve as setup 3 so it gives an almost identical accuracy, but it runs much slower because it requires almost 8 times more bands. 

\subsection{Scoring the reported matches}\label{sec:scoring}
When a \gls{PPRL} protocol relies on the Jaccard index but its implementation uses LSH, it may be in the users' interest to quantify the number of false-positive events. To this end, we present a way to estimate the Jaccard index based on the LSH results.

\subsubsection{Estimating the Jaccard index for matching pairs}\label{sec:scoring_prob}
When using LSH with $B$ band signatures, it is possible to estimate the actual Jaccard index $J$ by using a binomial confidence interval. By observation \ref{obs:mihashjac}, the probability for a matching band (i.e. the probability for a match in all $R$ Min-Hashes of the band) is $p=J^R$. Suppose that \party{\alice} learns that there are $h$ matching band signatures and $t=B-h$ non-matching band signatures. Using a $95\%$ confidence interval, the Jaccard index lies in the range 
\begin{align}\label{eq:95conf2}
    \left[
    \sqrt[R]{\left|\frac{h}{B} - 1.96 \sqrt{t \frac{h}{B^3}}\right|},
    \sqrt[R]{\left|\frac{h}{B} + 1.96 \sqrt{t \frac{h}{B^3}}\right|}
    \right]
\end{align}

In some cases this interval is too wide, and the users may prefer using a different approach, such as a revealing \gls{PPRL}.
In a revealing \gls{PPRL}, the two parties learn the intersection of their datasets as in a standard \gls{PPRL} but they also learn the records of the other parties that are involved in of the intersection. Thus, the leaked information in a revealing \gls{PPRL} is higher than in a \gls{PPRL}. Below, we propose an approach with privacy leakage that lies between the leakage of a revealing \gls{PPRL} and a \gls{PPRL}, where we compute the Jaccard index only for matching pairs, without revealing the exact shingles. 

\subsubsection{Computing the precise Jaccard index for matching pairs}\label{sec:scoring_prob}
Suppose that at the end of the LSH-PSI PPRL protocol, \party{\alice} learns the matching pair $(\aliceElement, Enc(\bobElement))$. \party{\alice} can ask \party{\bob} to participate in another \gls{PSI} process over the set of shingles of $(\aliceElement, \bobElement)$, where \party{\alice} knows \aliceElement and \party{\bob} knows \bobElement. In this \gls{PSI}, \party{\alice} only learns the intersection size of the associated shingles $|S \cap R|$ and the size $|R|$, so it can compute $J(\aliceElement, \bobElement) = \frac{|S \cap R|}{|S| + |R| - |S \cap R|}$. Note that learning only the intersection size and not the intersection itself makes it harder for \party{\alice} to guess \party{\bob}'s record. 

These additional \glspl{PSI} are relatively expensive in terms of performance, but we only need to carry them out for the reported matches, which are presumably only a very small fraction of all possible pairs of records. \party{\alice} and \party{\bob} can decide to perform such \glspl{PSI} for every matching pair or for selected pairs of special interest, or for pairs selected after estimating the Jaccard index as described above. As mentioned in Section \ref{sec:pprl} performing a selective \glspl{PSI} leaks the size of the selection to an eavesdropper \textbf{and this should be taken into account in the application threats model}.


\section{Our implementation}\label{sec:impl}

For reproducibility, we provide concrete details about our LSH implementation. We start by explaining the concept of relative weighting of the record fields.

\subsection{Relative weighting of the record fields}\label{sec:weights}
Some record fields may be more indicative of identity than other fields. 
For example, an \gls{SSN} field is very indicative (though it may also include typos), and a similarity of the full names is more indicative of identity than the similarity of zip codes. A simple method of weighting the effect of the different fields on the matching process is to duplicate the shingles originating from a field for a predefined number of times. We call this number the field weight. For example, consider a PPRL that operates over records with two fields: name and zip code. We use k=6 and k=7 shingles for these fields and set their weights to be 3 and 1, respectively. Then, the 6-shingle `John S' extracted from the name field `John Smith' will be duplicated into three separate shingles `John S1', `John S2', `John S3', whereas the zip-code 7-shingle `2304170' will not be duplicated. 
This causes shingles originating from the name to be three times more likely than zip-code shingles to be the minimum value used by the Min-Hashes of the LSH (see Section \ref{sec:LSH}). This will make the band signatures more likely to match if name shingles are identical than if zip-code shingles are identical.

The problem with this shingle duplication weighting method is that the extra shingles slow down the PPRL process because more shingles need to be hashed by the many Min-Hashes. To this end, we present  a novel method for weighting the shingles, which yields the same results as the shingle duplication method but is much faster. The idea is to reduce the hash value of a shingle according to the shingle's weight, to directly increase its chance of being the shingle that receives the minimal value by the Min-Hashes. 

We view the hash code $h$ of a shingle as a discrete random variable with uniform distribution over some integer range $[0,maxVal]$. Thus, $x=h/maxVal$ is approximately a random variable with a continuous uniform distribution over $[0,1]$. Our method relies on this being a good approximation. 

Our method is as follows: instead of duplicating a shingle $w$ times, we compute the shingle's hash-code $h$, normalize it $x=h/maxVal$, then apply the transformation $y=1-(1-x)^{1/w}$, and finally return back to the original scale $h'=\lfloor y *  maxVal\rfloor$. 
Lemma~\ref{lemma_weighted} shows that this results with a variable $h'$ whose distribution is the same as the minimum of $w$ independent hashes. 

\begin{lemma}
Let $H_1$, $H_2$, \ldots, $H_n$ be i.i.d. random variables with uniform distribution over $[0,1]$. Let $Y=min(H_1,H_2,\ldots,H_w)$.
Then $X=1-(1-H_1)^{1/w}$ has the same distribution as $Y$.
\label{lemma_weighted}
\end{lemma}

\begin{proof}
Let $F_H$ be the cumulative distribution function (CDF) of each $H_i$, i.e., $F_H(h)=h$ in the range $[0,1]$. 
Let $F_Y$ be the CDF of $Y$, i.e., $F_Y(y)=1-(1-F_H(y))^w=1-(1-y)^w$ and its inverse is $F_Y^{-1}(p)=1-(1-p)^{1/w}$, so $X=F_Y^{-1}(H_1)$.
The CDF of $X$ is therefore
\[
F_X(x)=P(X\leq x)=P(F_Y^{-1}(H_1)\leq x)=P(H_1\leq F_Y(x)).\]
Since $H_1$ is a uniform variable over $[0,1]$, this means $F_X(x)=F_Y(x)$.
\qed
\end{proof}

We observed a $9\%$ speedup when comparing the computation time (ignoring communications) of our \gls{PPRL} solution using the shingle duplication method versus the above hash-dropping method. 

\begin{remark}
The work in~\cite{Ioffe2010ImprovedCS} also describes a method of computing a `Weighted MinHash' over multisets with duplicated elements, but the universe of all possible items (or dimension for vectors) is assumed to be known in advance.
\end{remark}

\section{LSH description}

We are now ready to describe our \gls{LSH} implementation.
The algorithms below use a data structure that we call the field-group data structure $FG$, which is a list of tuples ($s, k, w$), where $s$ is a string, $k \in \N$ is the shingles length, and $w \in \N$ is a vector with the shingles' weights, respectively. Algorithm \ref{alg:lsh-REC} computes the LSH for a given record $record$. First, it concatenates together strings from fields that belong to the same group according to the configuration variable $conf$ (Lines 5-6). Then, it attaches to every concatenated string the $k,w$ values of its group as defined by $conf$ (Line 7). The algorithm returns the output of the \pcr{LshFG} function on the generated field-group data structure $FG$ (Line 8).

The \pcr{LshFG} algorithm uses the auxiliary functions \pcr{getWeigthedShingles}, which we describe in Algorithm \ref{alg:get-shingles}. Its input is a field-group data structure and its output is a list of pairs of $k$-shingles and their respective weights. 

\begin{algorithm}[ht!]
\caption{Compute the LSH for a given DB record}
\label{alg:lsh-REC}
\begin{algorithmic}[1]
    \Statex \textbf{Input:} $record$, a map of fields to values (strings) and $conf$ a list of tuples $(F, k, w)$ where $F$ is a set of field names, and $k, w \in \N$ are the shingles length and the fields weight, respectively.
    \Statex \textbf{Output:} $lsh=[b_1, b_2, \ldots, b_B]$.
    \Procedure{LSH}{$record$, $conf$}
        \State $FG = \emptyset$
        \For{$t \in conf$}
            \State $s =$``"
            \For{$f \in t.F$}
                \State $s = s ~|~ record[f]$
            \EndFor
            \State $FG = FG \cup (s, t.k, t.w)$
        \EndFor
        
        \State \Return $LshFG(FG)$
    \EndProcedure
\end{algorithmic}
\end{algorithm}

\begin{algorithm}[ht!]
\caption{Returns weighted shingles for given strings}
\label{alg:get-shingles}
\begin{algorithmic}[1]
\Statex \textbf{Input:} $FG$ a field-group data structure.
\Statex \textbf{Output:} $res$ an ordered list of pairs $(sh, w)$ where $sh$ is a string and $w \in \N$.
    \Procedure{getWeigthedShingles}{$FG$}
        \State $res = \emptyset$
        \For{$(s, k, w) \in FG$}
            \State $S$ = \pcr{getShingles}($s$, $k$)
            \Comment Returns an ordered list of the $k$-shingles of $s$.
            \State $res = res.append\left([(sh, w) ~|~ sh \in S]\right)$
        \EndFor
        \State \Return $res$
    \EndProcedure
\end{algorithmic}
\end{algorithm}

\begin{algorithm}[t!]
\caption{Compute the LSH for a given record field group}
\label{alg:lsh-FG}
\begin{algorithmic}[1]
    \Statex \textbf{Constants:} MP = $2^{61} - 1$, a Mersenne prime, and $maxVal = 2^{32}$
    \Statex \textbf{Input:} $h, c, d, w \in \N$.
    \Statex \textbf{Output:} an integer.
    \Procedure{CalcH}{$h, c, d, w$}
        \State $h = \left[h \cdot c + d \pmod{MP}\right] \pmod{maxVal}$
        \State \Return $maxVal \cdot \left(
                    1 - (1 - \frac{h}{maxVal})^\frac{1}{w}\right) $
        \Comment based on Lemma \ref{lemma_weighted}. 
    \EndProcedure
    \Statex

    \Statex \textbf{Input:} $B, R \in \N$, and $FG$ a field-group data structure.
    \Statex \textbf{Output:} $L = [b_1, b_2, \ldots, b_B]$.
    \Procedure{LshFG}{$B, R, FG$}
        \State $C \xleftarrow{\$} \{1,\ldots,MP\}^R$
        \State $D \xleftarrow{\$} \{0,\ldots,MP\}^R$
        \State $wS =$ \pcr{getWeightedShingles}($FG$)
        \For{$b=1, \ldots, B$}
            \State $i = 1$
            \For{$(sh, w) \in wS$}
                \State $H[i++] = ( $\pcr{Trunc}$_{32}($\pcr{SHA256}$(sh)), w)$
            \EndFor
            \For{$r=1,\ldots, R$}
                \State $M[r] = \min_{i}\{$\pcr{CalcH}$(H[i].sh, C[r], D[r], H[i].w)\}$
            \EndFor
            \State $L[b]$ = \pcr{SHA256}$(M)$
        \EndFor
        \State \Return $L$
    \EndProcedure
\end{algorithmic}
\end{algorithm}

Algorithm \ref{alg:lsh-FG} describes the function \pcr{LshFG}, which basically follows the LSH definition. First, the strings are converted to shingles by invoking Algorithm \ref{alg:get-shingles}. The loop of lines 8-14 generates the signature bands in $M$. It starts by computing a $32$-bit hash for every shingle (lines 10-11), and then uses them to construct $R$ different hashes for each of the shingles. The $R$ hash values are then reduced according to the shingle weight using the function \pcr{CalcH}. This function is based on Lemma \ref{lemma_weighted}, where the equation in Line 3 can be modified when $w \le 2$ to avoid the division and save computations. The resulting $R$ minimal hash values are kept in the $M$ array. To generate fast hash values, we replaced the intermediate \pcr{SHA256} calls with a Mersenne twister, which uses random numbers. The algorithm generates and holds these numbers in the arrays $C$ and $D$. Finally, using \pcr{SHA256}, we concatenate and hash the values of $M$ to create the band signature (Line 14).

\end{document}